\documentclass{article}
\usepackage{amssymb}

\usepackage{graphicx}
\usepackage{amsmath}
\usepackage{mathrsfs}

\newtheorem{theorem}{Theorem}
\newtheorem{acknowledgement}[theorem]{Acknowledgement}

\newtheorem{definition}[theorem]{Definition}

\newtheorem{lemma}[theorem]{Lemma}

\newenvironment{proof}[1][Proof]{\textbf{#1.} }{\ \rule{0.5em}{0.5em}}

\begin{document}

\title{A Trotter-Kato Theorem for Quantum Markov Limits}
\author{Luc Bouten\footnote{\texttt{luc\_bouten@hotmail.com}, BTA VOF, Slootsekuilen 9, 5986 PE, Beringe, The Netherlands},
Rolf Gohm\footnote{\texttt{rog@aber.ac.uk}
Dept. for Mathematics and Physics, Aberystwyth University, SY23 3BZ, Wales, United Kingdom}, 
John Gough\footnote{\texttt{jug@aber.ac.uk}
Dept. for Mathematics and Physics, Aberystwyth University, SY23 3BZ, Wales, United Kingdom} and Hendra Nurdin\footnote{
 \texttt{h.nurdin@unsw.edu.au}
 School of Electrical Engineering and Telecommunications, UNSW Australia, Sydney NSW 2052, Australia.}
 }
 
\maketitle

\begin{abstract}
Using the Trotter-Kato theorem we prove the convergence of the unitary dynamics generated by
an increasingly singular Hamiltonian in the case of a single field coupling. The limit dynamics is 
a quantum stochastic evolution of Hudson-Parthasarathy type, and we establish in the process a 
graph limit convergence of the pre-limit Hamiltonian operators to the Chebotarev-Gregoratti-von Waldenfels
Hamiltonian generating the quantum It\={o} evolution.
\end{abstract}

\section{Introduction}

In the situation of regular perturbation theory, we typically have a
Hamiltonian interaction of the form $H=H_{0}+H_{\text{int}}$ with associated
strongly continuous one-parameter unitary groups $U_{0}\left( t\right)
=e^{-itH_{0}}$ (the free evolution) and $U\left( t\right) =e^{-itH}$ (the
perturbed evolution), then we transform to the Dirac interaction picture by
means of the unitary family $V\left( t\right) =U_{0}\left( -t\right) U\left(
t\right) $. Although $V\left( \cdot \right) $ is strongly continuous, it
does not form a one-parameter group but instead yields what is known as a
left $U_{0}$-cocycle: 
\begin{equation}
V\left( t+s\right) =U_{0}\left( s\right) ^{\dag }V\left( t\right)
U_{0}\left( s\right) V\left( s\right) .  \label{eg:cocyle}
\end{equation}
One obtains the interaction picture dynamical equation 
\begin{equation}
i\dfrac{d}{dt}V\left( t\right) =\Upsilon \left( t\right) V\left( t\right) ,
\label{eq:Dirac}
\end{equation}
where $\Upsilon \left( t\right) =U_{0}\left( t\right) ^{\dag }H_{\text{int}
}U_{0}\left( t\right) $.

More generally, we may have a pair of unitary groups $U\left( \cdot \right) $
and $U_{0}\left( \cdot \right) $ with Stone generators $H$ and $H_{0}$
respectively, but where the intersection of the domains of the generators are not dense. This is
the situation of a singular perturbation. In this case we cannot expect the
Dirac picture dynamical equation (\ref{eq:Dirac}) to be anything but formal
since the difference $H_{\text{int}}=H-H_{0}$ is not densely defined.

Remarkably, the steps above can be reversed even for the situation of singular perturbations. If we assume at the outset a fixed free dynamics $U_{0}\left( \cdot \right) $, with Stone generator $H_{0}$, and a strongly continuous unitary left $U_{0}$-cocycle $V\left( \cdot \right) $, then $U\left( t\right) =U_{0}\left( t\right) V\left( t\right) $ will then form a
strongly continuous one-parameter unitary group with Stone generator $H$. In practice however the problem of reconstructing $H$ from the prescribed $H_{0} $ and $V\left( \cdot \right) $ will be difficult.

In the situation of quantum stochastic evolutions introduced by Hudson and Parthasarathy \cite{HP84}, we have a strongly continuous adapted process $V\left( \cdot \right) $ satisfying a quantum stochastic differential equation (including Wiener and Poisson noise as special commutative cases) in place of (\ref{eq:Dirac}), and the solution constitutes a cocycle with
respect to the time-shift maps $U_{0}\equiv \Theta $ (see below). Nevertheless, $V\left( \cdot \right) $ arises as the Dirac picture evolution for a singular perturbation of a unitary $U\left( \cdot \right) $ with some generator $H$ with respect to the time-shift: it was a long standing problem to find an explicit form for $H$ which was finally resolved by Gregoratti \cite{Gre01}, see also \cite{WvW}.

The purpose of this paper is to approximate the singular perturbation arising in quantum stochastic evolution models by a sequence of regular perturbation models. That is, to construct a sequence of Hamiltonians $H^{\left( k\right) }=H_{0}+H_{\text{int}}^{\left( k\right) }$ yielding a regular perturbation $V^{\left( k\right) }\left( \cdot \right) $ converging
to a singular perturbation $V\left( \cdot \right) $ in some controlled way. We exploit the fact that the limit Hamiltonian is now known through the work of Chebotarev \cite{Che97} and Gregoratti \cite{Gre01}. The strategy is to employ the Trotter-Kato theorem which guarantees strong uniform convergence of the unitaries once graph convergence of the Hamiltonians is established.

\subsection{Quantum Stochastic Evolutions}

The seminal work of Hudson and Parthasarathy \cite{HP84} on quantum stochastic evolutions lead to explicit constructions of unitary adapted quantum stochastic processes $V$ describing the the open dynamical evolution of a system with a singular Boson field environment. We fix the system Hilbert space $\mathfrak{h}$ and model the environment as having $n$ channels so that the underlying Fock space is $\mathfrak{F}=\Gamma \left( \mathbb{C}^n \otimes L^{2} ( \mathbb{R}) \right) $. Here $\Gamma \left( \mathfrak{H}\right) $ denotes the symmetric (boson) Fock space over a one-particle space $\mathfrak{H}$: we set the inner product as $\langle \Psi |\Phi \rangle =\sum_{m=0}^{\infty }\frac{1}{m!}\langle \Psi _{m}|\Phi _{m}\rangle $ and take the exponential vectors to be defined as  ($\otimes_s$ denoting a symmetric tensor product)
\begin{equation*}
e\left( f\right) =\left( 1,f,f\otimes_s f,f\otimes_s f\otimes_s f,\cdots \right)
\end{equation*}
with test function $f\in \mathfrak{H}$. Here the one particle space is $L^{2}(\mathbb{R})$, the space of complex-valued square-integrable functions on $\mathbb{R}$. We define the operators  
\begin{eqnarray*}
\Lambda ^{00}\left( t\right) &\triangleq &t, \\
\Lambda ^{10}\left( t\right) &=&A^{\dag }\left( t\right) \triangleq
a^{\dag }\left( 1_{\left[ 0,t\right] }\right) , \\
\Lambda ^{01}\left( t\right) &=&A\left( t\right) \triangleq a\left(
 1_{\left[ 0,t\right] }\right) , \\
\Lambda ^{11}\left( t\right) &= & \Lambda (t) \triangleq d\Gamma \left( \chi _{\left[ 0,t\right] }\right) ,
\end{eqnarray*}
where $1_{\left[ 0,t\right] }$ is the characteristic function of the interval $\left[ 0,t\right] $ and $\chi _{\left[ 0,t\right] }$ is the operator on $L^{2}(\mathbb{R}) $ corresponding to multiplication by $1_{\left[ 0,t\right] }$. Hudson and Parthasarathy \cite{HP84} have developed a quantum It\={o} calculus where the basic objects are integrals of adapted processes with respect to the fundamental
processes $\Lambda ^{\alpha \beta }$. The quantum It\={o} table is then 
\begin{equation*}
d\Lambda ^{\alpha \beta }\left( t\right) \,d\Lambda ^{\mu \nu }\left(
t\right) =\hat{\delta}_{\beta \mu }\,d\Lambda ^{\alpha \nu }\left( t\right)
\end{equation*}
where $\hat{\delta}_{\alpha \beta }$ is the Evans-Hudson delta defined to equal unity if $\alpha =\beta =1 $ and zero otherwise. This may be written as 
\begin{equation*}
\begin{tabular}{l|llll}
$\times $ & $dA$ & $d\Lambda $ & $dA^{\dag }$ & $dt$ \\ \hline
$dA$ & 0 & $dA$ & $dt$ & 0 \\ 
$d\Lambda $ & 0 & $d\Lambda $ & $dA$ & 
0 \\ 
$dA^{\dag }$ & 0 & 0 & 0 & 0 \\ 
$dt$ & 0 & 0 & 0 & 0
\end{tabular}
.
\end{equation*}
In particular, we have the following theorem \cite{HP84}.

\begin{theorem}
There exists a unique solution $V\left( \cdot ,\cdot \right) $ to the quantum stochastic differential equation 
\begin{equation}
V\left( t,s\right) =I+\int_{s}^{t}dG\left( \tau \right) \,V\left( \tau
,s\right)  \label{eq: diff eq}
\end{equation}
$\left( t\geq s\geq 0\right) $ where 
\begin{equation*}
dG\left( t\right) =G_{\alpha \beta }\otimes d\Lambda ^{\alpha \beta }\left(
t\right)
\end{equation*}
with $G_{\alpha \beta }\in \mathfrak{B}\left( \mathfrak{h}\right) $. (We adopt the convention that we sum repeated Greek indices over the range $0,1$.)
\end{theorem}

In particular, set $V\left( t\right) =V\left( t,0\right) $ then we have the quantum stochastic differential equation $dV\left( t\right) =dG\left( t\right) \,V\left( t\right) $ which replaces the regular Dirac picture dynamical equation (\ref{eq:Dirac}). 

We refer to $\mathbf{G=}\left[ G_{\alpha \beta }\right] \in \mathfrak{B} \left( \mathfrak{h} \oplus \mathfrak{h}  \right) $, as the \textit{coefficient matrix}, and $V$ as the left process generated by $\mathbf{G}$. The conditions for the process $V$ to be unitary are that $\mathbf{G}$\ takes the form, with respect to the decomposition $\mathfrak{h} \oplus \mathfrak{h} $, 
\begin{equation}
\mathbf{G}=\left[ 
\begin{array}{cc}
-\dfrac{1}{2}\mathsf{L}^{\dag }\mathsf{L}-i\mathsf{H} & -\mathsf{L}^{\dag } \mathsf{S} \\ 
\mathsf{L} & \mathsf{S}-I
\end{array}
\right]  \label{eq:G_unitary}
\end{equation}
where $\mathsf{S}\in \mathfrak{B}\left( \mathfrak{h}\right) $ is a unitary, $\mathsf{L}
\in \mathfrak{B}\left( \mathfrak{h}\right) $ and $\mathsf{H}\in \mathfrak{B}\left( \mathfrak{h}\right) $ is self-adjoint. We may write in more familiar notation \cite{HP84} 
\begin{equation*}
dG\left( t\right) =\left( -\dfrac{1}{2}\mathsf{L}^{\dag }\mathsf{L}-i \mathsf{H}\right) \otimes dt-\mathsf{L}^{\dag }\mathsf{S}\otimes
dA\left( t\right) +\mathsf{L}\otimes dA^{\dag }\left( t\right) +(\mathsf{S}-I)\otimes d\Lambda \left( t\right) .
\end{equation*}

We denote the shift map on $L^{2}\left( \mathbb{R} \right) $ by $\theta _{t}$, that is $\left( \theta _{t}\right) f(\cdot)=f\left( \cdot +t\right) $ and its second quantization as $\Theta _{t}=I\otimes \Gamma \left( \theta _{t}\right) $. It then turns out that $\Theta _{\tau }^{\dag }V\left( t,s\right) \Theta _{\tau }=V\left( t+\tau ,s+\tau \right) $ and so $V\left( t\right) =V\left( 0,t\right) $ is a left unitary $\Theta $-cocycle and that there must exist a self-adjoint operator $H$ such that  
\begin{equation*}
\Theta _{t}V\left( t\right) \equiv e^{-iHt}
\end{equation*}
for $t\geq 0$. (For $t<0$ one has $V\left( -t\right) ^{\dag }\Theta
_{-t}\equiv e^{-iHt}$.) Here $H$ will be a singular perturbation of generator of the shift, and its characterization was given by Gregoratti \cite{Gre01}. See also \cite{QG}.

\subsection{Physical Motivation}

As a precursor to and motivation for further approximations, we fix on a simple model of a quantum mechanical system coupled to a boson field reservoir $R$. In the Markov approximation we assume that the auto-correlation time of the field processes vanishes in the limit: this includes weak coupling (van Hove) and low density limits. The Hilbert space for the field is the Fock space $\mathcal{F}_{R}=\Gamma \left( \mathcal{H}_{R}^{1}\right) $ with one-particle space $\mathcal{H}_{R}^{1}=L^{2}\left(  \mathbb{R}\right) $ taken as the momentum space. (For convenience we consider a one-dimensional situation because this is the setting studied in this paper but of course $\mathbb{R}^3$ is particularly relevant physically.)
It is convenient to write annihilation operators formally as $A_{R}\left( g\right) =\int_{\mathbb{R}}g\left( p\right) ^{\ast }a_{p}dp$ where $\left[ a_{p},a_{p^{\prime }}^{\dag }\right] =\delta \left( p-p^{\prime }\right) $.

In particular, let us fix a function $g\in L^{2}\left(  \mathbb{R}\right) $, and set 
\begin{equation*}
a\left( t,k\right) =\sqrt{k}\int e^{-i\omega \left( p\right)
tk}g\left( p\right) ^{\ast }a_{p}\,dp
\end{equation*}
where $\omega =\omega \left( p\right) $ is a given function (determining the dispersion relation for the free quanta) and $k$ is a dimensionless parameter rescaling time. We have the commutation relations 
\begin{equation*}
\left[ a\left( s,k\right) ,a\left( t,k\right) ^{\dag }\right]
=k\, \rho \left( k(t-s)\right)
\end{equation*}
where 
\begin{equation*}
\rho \left( \tau \right) \equiv \int  \vert g \left( p\right) \vert^2 e^{i\omega \left( p\right) \tau }\,dp.
\end{equation*}
The limit $k\rightarrow \infty $ leads to singular commutation relations, and it is convenient to introduce smeared fields 
\begin{equation*}
A\left( \varphi ,k\right) =\int \varphi \left( t\right) ^{\ast}a \left( t,k\right) \,dt
\end{equation*}
in which case we have the two-point function (and define an operator $C_k$ by)
\begin{equation*}
\left[ A\left( \varphi ,k\right) ,A\left( \psi ,k\right) ^{\dag }\right]
=\int dtdt^{\prime }\, \varphi \left( t\right) ^{\ast }k\rho \left( k\left( t-t^{\prime }\right) \right) \psi \left( t^{\prime
}\right) \equiv \langle \varphi |C_{k}\psi \rangle
\end{equation*}
For $\rho $ integrable, we expect 
\begin{equation*}
\lim_{k\rightarrow \infty }\left[ A\left( \varphi ,k\right) ,A\left( \psi
,k\right) ^{\dag }\right] = \gamma \int dt\, \varphi \left( t\right)
^{\ast } \psi \left( t\right) 
\end{equation*}
where $\gamma =\int_{-\infty }^{\infty }\rho \left( \tau \right) d\tau =2\pi \int \vert g \vert^2 \left( p\right) \delta \left( \omega \left( p\right) \right) \,dp \ge 0$. When $\gamma = 1$, the $A\left( \varphi ,k\right) $ are smeared versions of the annihilators on $\Gamma \left( L^{2}\left( \mathbb{R}\right) \right) $.

The limit $k \uparrow \infty$ corresponds to the smeared field becoming singular and this leads to a quantum Markovian approximation.
The formulation of such models was first given and treated in a systematic way by Accardi, Frigero and Lu who developed a
set of powerful quantum functional central limit theorems including the weak coupling \cite{AFL1} and low density \cite{AFL2} regimes.
Theorem \ref{thm:JG} is an extension of these which includes both quantum diffusion and jump terms \cite{Go05,Go06}.

\begin{theorem}
\label{thm:JG}
Let $\left( \mathsf{E}_{\alpha \beta }\right) $ be bounded operators on a
fixed separable Hilbert space $\mathfrak{h}$ labeled by $\alpha ,\beta \in
\left\{ 0,1 \right\} $ with $\mathsf{E}_{\alpha \beta }^{\dag }=
\mathsf{E}_{\beta \alpha }$ and $ \| E_{11} \| < 2$. Let 
\begin{equation*}
\Upsilon \left( t,k\right) =\mathsf{E}_{11}\otimes a\left(
t,k\right) ^{\dag }a\left( t,k\right) +\mathsf{E}_{10}\otimes
a\left( t,k\right) ^{\dag }+\mathsf{E}_{01}\otimes a\left(
t,k\right) +\mathsf{E}_{00}\otimes I
\end{equation*}
and 
\begin{equation*}
e\left( \varphi ,k\right) =\exp \left\{ A\left( \varphi ,k\right) -A\left(
\varphi ,k\right) ^{\dag }\right\} \Omega _{R}
\end{equation*}
with $\Omega _{R}$ the Fock vacuum of $\mathcal{F}_{R}$. The solution $V (t,k)$ to the equation 
\begin{equation*}
\frac{d}{dt}V\left( t,k\right) =-i\Upsilon \left( t,k\right) \,V\left(
t,k\right) ,\quad V\left( 0,k\right) =I,
\end{equation*}
exists and we have the limit 
\begin{equation*}
\lim_{k\rightarrow \infty }\langle u_{1}\otimes e\left( \varphi ,k\right)
|V\left( t,k\right) |u_{2}\otimes e\left( \psi ,k\right) \rangle =\langle
u_{1}\otimes e\left( \varphi \right) |V\left( t\right) |u_{2}\otimes e\left(
\psi \right) \rangle
\end{equation*}
for all $u_{1},u_{2}\in \mathfrak{h}$ and $\varphi ,\psi \in L^{2}\left( \mathbb{R}\right) $, where $V$ is a unitary
adapted process on $\mathfrak{h}\otimes \Gamma \left( \mathbb{C}^{n}\otimes
L^{2}\left( \mathbb{R}\right) \right) $ with coefficient matrix $\mathbf{G%
}$  given by 
\begin{equation}
\mathbf{G}=-i\mathbf{E}-i\frac{1}{2} \mathbf{G}\left[ 
\begin{array}{cc}
0 & 0 \\ 
0 & 1
\end{array}
\right] \mathbf{E}  \label{eq: Ito to Stratonovich}
\end{equation}
where  we assume $ \int_{-\infty }^{0}\rho \left( \tau \right) d\tau = \frac{1}{2}.$
\end{theorem}

The proof of the Theorem is given in \cite{Go05} and requires a development and a uniform estimation of the Dyson series expansion. Summability of the series requires
that $\| E_{11} \| < 2$.

The triple $(\mathsf{S},\mathsf{L},\mathsf{H})$ from (\ref{eq:G_unitary}) obtained through (\ref{eq: Ito to Stratonovich}) is 
\begin{eqnarray}
\mathsf{S} &=&\frac{I- \frac{i}{2} \mathsf{E}_{11}}{I+ \frac{i}{2} \mathsf{E}_{11}},\qquad 
\mathsf{L}=-\frac{i}{I+ \frac{i}{2}  \mathsf{E}_{11}} \mathsf{E} _{10},  \nonumber\\
\mathsf{H} &=&\mathsf{E}_{00}+\mathsf{E}_{01} \mathrm{Im} \left\{ \frac{1}{I+ \frac{i}{2} \mathsf{E}_{11}}\right\} \mathsf{E}_{10}.  \label{eq: SLH 1D}
\end{eqnarray}

Our objective is reappraise Theorem \ref{thm:JG}, where we will prove a related result by an alternative technique.
Using the Trotter-Kato theorem, we will establish a stronger mode of convergence (uniformly on compact intervals of time
and strongly in the Hilbert space) by means of a graph convergence of the Hamiltonians. The new approach has the advantage 
of been simpler and is likely to be more readily extended to other cases, for instance a continuum of input channels as 
originally treated in \cite{HP84}, which cannot be treated by the perturbative techniques used in the proof of Theorem \ref{thm:JG}.

\section{Trotter-Kato Theorems for Quantum \newline
Stochastic Limits}

Our main results will employ the Trotter-Kato theorem, which we recall next in a particularly convenient form. See \cite{Da}, Theorem 3.17, or \cite{RSI}, Chapter VIII.7. 

\begin{theorem}[\label{thm Trotter-Kato}Trotter-Kato]
Let $\mathcal{H}$ be a Hilbert space and let $U^{\left( k\right) }\left( \cdot \right) $ and $U\left( \cdot \right) $ be strongly continuous one-parameter groups of unitaries on $\mathcal{H}$ with Stone generators $H^{(k)}$ and $H$, respectively. Let $\mathcal{D}$ be a core for $H$. The following are equivalent

\begin{enumerate}
\item  For all $f\in \mathcal{D}$ there exist $f^{(k)}\in \mbox{Dom} (H^{(k)})$ such that 
\begin{equation*}
\lim_{k\rightarrow \infty }f^{(k)}=f,\qquad \lim_{k\rightarrow
\infty}H^{(k)}f^{(k)}=Hf.
\end{equation*}

\item  For all $0\leq T<\infty $ and all $f\in \mathcal{H}$ we have 
\begin{equation*}
\lim_{k\rightarrow \infty }\sup_{0\leq t\leq T}\left\| \left( U^{\left(
k\right) }\left( t\right) -U\left( t\right) \right) f\right\| =0.
\end{equation*}
\end{enumerate}
\end{theorem}

The theorem yields a strong uniform convergence if we can establish graph convergence of the Hamiltonians. We now present the Trotter-Kato theorems for the class of problems that interest us, treating the first and second quantized problems in sequence.

\subsection{First Quantization Example}

\begin{definition}
\label{def:gk} Let $g \in C_{c}^{\infty}(\mathbb{R})$, i.e., an infinitely differentiable function with compact support, such that $\int_{-\infty }^{\infty } g(s)ds=1$. We define $\rho (t)=\int_{\mathbb{R} }g(s)^{\ast}\, g(s+t)ds$. Moreover, for all $k>0$, we define functions $g^{(k)} $ and $\rho ^{(k)}$ by 
\begin{equation*}
g^{(k)}(t)=k \,g(kt),\qquad \rho ^{(k)}(t)=k \,\rho (kt),\qquad t\in \mathbb{R}.
\end{equation*}
Furthermore, we define two complex numbers by $\kappa _{+} :=\int_{0}^{\infty }\rho (s)ds$ and $\kappa _{-} :=\int_{-\infty }^{0}\rho (s)ds$.
\end{definition}

Note that $\kappa _{+} + \kappa _{-} =1$ and that $\kappa _{+}$ and $\kappa _{-}$ are complex conjugate: $\kappa _{+}=(\kappa _{-})^{\ast}$ (substitute $-s$ for $s$), hence $\kappa _{\pm}=\frac{1}{2} \pm i\sigma $ with $\sigma$ real.  
The choice of $\rho$ is such that $\langle g \vert g \ast f \rangle = \langle \rho \vert f \rangle$, where
$(g \ast f) (t) = \int_{-\infty }^{\infty } g(s) f(t-s) ds$ is the usual convolution.

Let $\mathfrak{h}$ be a Hilbert space and let $\mathsf{E}$ be a bounded self-adjoint operator on $\mathfrak{h}$. We consider the following family of operators on 
$L^{2}(\mathbb{R};\mathfrak{h}) \simeq \mathfrak{h} \otimes L^{2}(\mathbb{R})$: 
\begin{eqnarray}
H^{(k)} &=&i \,\partial +\mathsf{E}\,|g^{(k)}\rangle \langle g^{(k)}|
\simeq I \otimes i \,\partial + E \otimes |g^{(k)}\rangle \langle g^{(k)}|,  \notag \\
\mbox{Dom}(H^{(k)}) &=&W^{1,2}(\mathbb{R};\mathfrak{h}),  \label{eq: Hk}
\end{eqnarray}
where $W^{1,2}(X;\mathfrak{h})$, $X \subseteq \mathbb{R}$, denotes the Sobolev space of $\mathfrak{h}$-valued functions square integrable on $X$ with square integrable weak derivatives on $X$. It follows easily that $H^{(k)}$ is self-adjoint for every $k>0$ (for example by the Kato-Rellich theorem, see 
\cite{RSII}, Theorem X.12). We define a unitary operator on $\mathfrak{h}$ by 
\begin{equation}
\mathsf{S}=\dfrac{I-i\kappa _{-}\mathsf{E}}{I+i\kappa _{+}\mathsf{E}}.
\label{eq: S}
\end{equation}
and an operator $H$ on $L^{2}(\mathbb{R};\mathfrak{h})$ by 
\begin{eqnarray}
\mbox{Dom}(H) &=&\left\{ f\in W^{1,2}\left( \mathbb{R}\backslash \{0\};\mathfrak{\ h} \right) :f(0^{-})=\mathsf{S}f(0^{+})\right\} ,  \notag \\
Hf &=&i \,\partial f.
\end{eqnarray}
It follows easily that $H$ is self-adjoint, compare \cite{RSI}, VIII.2, final example. 

Remark: Any $f\in W^{1,2}\left( \mathbb{R}\backslash \{0\};\mathfrak{\ h} \right)$ is absolutely continuous both on $(-\infty,0)$ and $(0,\infty)$, see for example \cite{Gu}, 2.6 Ex.6, but the exclusion
of test functions supported at $0$ allows jumps at $0$. Higher dimensional situations ($\mathbb{R}^n$ with
$n>1$) are more complicated in this respect.  

We define strongly continuous one-parameter groups of unitaries on $L^{2}(\mathbb{R}; \mathfrak{h})$ by
\begin{equation*}
U^{\left( k\right) }\left( t\right) =\exp (-itH^{(k)}),\qquad U\left(
t\right) =\exp (-itH).
\end{equation*}

We then have the following theorem.

\begin{theorem}
\label{thm main theorem} Let $0\leq T<\infty $. Then 
\begin{equation*}
\lim_{k\rightarrow \infty }\sup_{0\leq t\leq T}\left\| \left( U^{\left(
k\right) }\left( t\right) -U\left( t\right) \right) f\right\| =0,\qquad
\forall f\in L^{2}(\mathbb{R};\mathfrak{h}).
\end{equation*}
\end{theorem}

We prove Theorem \ref{thm main theorem} at the end of this subsection. From the Trotter-Kato Theorem \ref{thm Trotter-Kato}, it suffices to find, for every $f\in \mbox{Dom}(H)$, a sequence $f^{(k)}\in \mbox{Dom}(H^{(k)})$ that satisfies condition (i) of Theorem \ref{thm Trotter-Kato}.

If $g$ is a $\mathbb{C}$-valued function on $X$ and $f \in L^2(X;\mathfrak{h}) \simeq \mathfrak{h} \otimes L^2(X;\mathbb{C})$ then we use the short notation $gf$ for $(I \otimes M_g)\,f$ where $M_g$ is multiplication by $g$. With this convention we can also define $g * f \in L^2(X;\mathfrak{h})$ and $\langle g  | f \rangle \in \mathfrak{h}$ for suitable functions $g$, using the same formulas as for $\mathfrak{h} = \mathbb{C}$. 

\begin{definition}
\label{def fk} Let $f$ be an element in the domain of $H$. Define an element $f^{(k)}$ in the domain of $H^{(k)}$ by 
\begin{equation*}
f^{(k)}(t) = (g^{(k)}*f)(t) = \int_{-\infty}^\infty g^{(k)}(t-s)f(s)ds.
\end{equation*}
\end{definition}

\begin{lemma}
\label{lem convergefast}Let $\eta $ be an element of $C(0,\infty )$ with compact support and let $h$ be an element of $W^{1,2}((0,\infty);\mathfrak{h})\,\cap\, C^{1}((0,\infty);\mathfrak{h})$ such that $h(0^{+})=0$. Let $\eta ^{(k)}(x)=k \,\eta (kx)$ for all $x\in (0,\infty )$ and $k>0$. Then 
\begin{equation*}
\|\langle \eta ^{(k)}|h\rangle \|_2 \leq \dfrac{C}{k},\qquad \forall k>0,
\end{equation*}
for some positive constant $C$.
\end{lemma}

\begin{proof}
Note that the $C^{1}$-function $h$ is Lipschitz on the support of $\eta$, that is, there exists a positive constant $L$ such that 
\begin{equation*}
\|h(x)-h(y)\|_2 \leq L|x-y|,\qquad \forall x,y\in \mathrm{supp}(\eta),
\end{equation*}
where $\mathrm{supp}(\eta)$ denotes the support of $\eta$. Taking the limit for $y$ to $0^{+}$ gives 
\begin{equation*}
\|h(x)\|_2\leq L|x|,\qquad x\in \mathrm{supp}(\eta).
\end{equation*}
We can define $M:=\max_{x\in (0,\infty )}|\eta (x)|$ and let $N$ be a number to the right of the support of $\eta $. Now we have 
\begin{eqnarray*}
\|\langle \eta ^{(k)}|h\rangle \|_2 \leq k\int_{0}^{\infty }|\eta
(kx)|\,\|h(x)\|_2 \,dx\\ \leq \dfrac{L}{k}\int_{0}^{\infty }|\eta (u)|\,u \,du\leq \dfrac{L}{k} \int_{0}^{N}Mu \,du 
=\dfrac{LMN^{2}}{2k}.
\end{eqnarray*}
\hfill
\end{proof}

\begin{lemma}
\label{lem:TKconditions} If $f$ is in $\mbox{Dom}(H)\cap C^\infty(\mathbb{R} \backslash \{0\};\mathfrak{h})$, and $f^{(k)}$ is given by Definition \ref{def fk}, then
we have 
\begin{equation*}
\begin{split}
1.\ \ \lim_{k\to \infty}\left\| f^{(k)} - f \right\|_2 = 0,\qquad\qquad 2.\
\ \lim_{k\to \infty}\left\| H^{(k)}f^{(k)} - Hf\right\|_2 = 0.
\end{split}
\end{equation*}
\end{lemma}

\begin{proof}
Note that the first limit follows immediately from a standard result on approximations by convolutions, see e.g.\ \cite[Thm.\ 2.16]{LiL97}. For the second limit, note that 
\begin{equation}
\partial (g^{(k)}\ast f)=g^{(k)}\ast \partial f + (f(0^{+})-f(0^{-}))g^{(k)},  \label{eq: deriv 1st}
\end{equation}
% where $\partial_{ac} f$ denotes the absolutely continuous part of the distributional derivative of $f$. 
Because $\partial f = Hf$ and using \cite[Thm.\ 2.16]{LiL97} once more, we find that 
\begin{equation*}
\lim_{k\rightarrow \infty }g^{(k)}\ast Hf=Hf.
\end{equation*}
That is, all we need to show is that 
\begin{equation}
\lim_{k\rightarrow \infty }\left\| \left( if(0^{+})-if(0^{-})+\mathsf{E}
\langle g^{(k)}|g^{(k)}\ast f\rangle \right) g^{(k)}\right\| _{2}=0.
\label{eq: to show}
\end{equation}
Note that $\langle g^{(k)}|\,g^{(k)}\ast f\rangle =\langle \rho^{(k)}|\,f\rangle $. We can now apply Lemma \ref{lem convergefast} with $ h=f\chi _{(0,\infty )}-f(0^{+})$ and $\eta =\rho \chi _{(0,\infty )}$ (resp.\ $h=f\chi _{(-\infty ,0)}-f(0^{-})$ and $\eta =\rho \chi _{(-\infty ,0)}$) to conclude that 
\begin{equation*}
\langle \rho ^{(k)}|\,f\rangle \overset{k\rightarrow \infty }{%
\longrightarrow }\left( \kappa _{-}\right) ^{\ast }f(0^{-})+\left( \kappa
_{+}\right) ^{\ast }f(0^{+})=\kappa _{+}f(0^{-})+\kappa _{-}f(0^{+}),
\end{equation*}
with rate $\dfrac{1}{k}$. Using the boundary condition for $f$, we therefore find that 
\begin{equation*}
if(0^{+})-if(0^{-})+\mathsf{E}\langle g^{(k)}|g^{(k)}\ast f\rangle \longrightarrow
i \big[ (I-i\kappa _{-}\mathsf{E})f(0^{+}) - (I+i\kappa _{+}\mathsf{E})f(0^{-}) \big]
= 0,
\end{equation*}
with rate $\dfrac{1}{k}$. Note that the $L^{2}$-norm of $g^{(k)}$ grows with rate $\sqrt{k}$, so that the limit in Eq.\ \eqref{eq: to show} follows. This completes the proof of the Lemma.
\end{proof}

\bigskip

\begin{proof}
\textbf{[of Theorem \ref{thm main theorem}]} The Theorem follows from a combination of the results in Theorem \ref{thm Trotter-Kato} and Lemma \ref{lem:TKconditions} and the fact that $\mbox{Dom}(H)\cap C^{\infty }(\mathbb{R} \backslash \{0\};\mathfrak{h})$ is a core for $H$. The latter follows from \cite[Thm. 7.6]{LiL97}.
\end{proof}

\section{A Second Quantized Model}

Let $\mathsf{E}_{\alpha \beta }$ be bounded operators on $\mathfrak{h}$ such that $\mathsf{E}_{\alpha \beta }^{\dag }=\mathsf{E}_{\beta \alpha }$ for $\alpha ,\beta \in \left\{ 0,1 \right\} $. Consider the following family of operators on $\mathfrak{h}\otimes \mathcal{F}$ 
\begin{equation}
H^{(k)}=id\Gamma (\partial )+ \mathsf{E}_{11}A^{\dag
}(g^{(k)})A(g^{(k)})+ \mathsf{E}_{10}A^{\dag
}(g^{(k)})+ \mathsf{E}_{01}A(g^{(k)})+\mathsf{E}_{00},
\end{equation}
choosing a suitable domain $\mbox{Dom} \left( H^{(k)}\right) $ of essential self-adjointness for all $k>0$. (We conjecture that
$\mathfrak{h}\otimes \mathcal{E}(C^\infty_c(\mathbb{R}))$, where
$\mathcal{E}(C^\infty_c(\mathbb{R}))$ is the set of exponential vectors $e(f)$ with $f \in C^\infty_c(\mathbb{R})$, is a set of analytic
vectors for the $H^{(k)}$ but we haven't been able to prove this rigorously and leave it as an open problem.)

We denote the strongly continuous group of unitaries on $\mathfrak{h}\otimes \mathcal{F}$ generated by the unique self-adjoint extension of $H^{(k)}$ by $U^{\left( k\right) }\left( t\right) $. Let the triple $(\mathsf{S},\mathsf{L},\mathsf{H})$ appearing in (\ref{eq:G_unitary}) be obtained from $\mathbf{E}=\left(  \mathsf{E}_{\alpha \beta }\right) $ through (\ref{eq: Ito to Stratonovich}):
see (\ref{eq: SLH 1D}).

The space $\mathfrak{h}\otimes \mathcal{F}=\mathfrak{h}\otimes \Gamma \left(  L^{2}\left( \mathbb{R}\right) \right) $ consists of vectors $\Psi = \left( \Psi _{m}\right) _{m\geq 0}$ which are sequences of symmetric $\mathfrak{{h}-}$valued functions $\Psi _{m}\left( t_{1},\cdots ,t_{m}\right) $ where $t_{j} \in\mathbb{R}$. Following Gregoratti \cite{Gre01}, we define the following spaces: (for $I$ a Borel subset of $\mathbb{R}$ and $\mathfrak{H}$ a Hilbert space)

\begin{eqnarray*}
\mathscr{H}^{\Sigma }\left( I^{m},\mathfrak{H}\right) &=&\left\{ v\in
L^{2}\left( I^{m},\mathfrak{H}\right) :\sum_{i=1}^{m}\partial _{i}v\in
L^{2}\left( I^{m},\mathfrak{H}\right) \right\} ; \\
\mathscr{W} &=&\left\{ 
\begin{array}{c}
\Psi \in \mathfrak{h}\otimes \mathcal{F}:\Psi _{m}\in \mathscr{H}^{\Sigma
}\left(   \mathbb{R} ^{m},%
\mathfrak{h}\right) : \\ 
\sum_{m=0}^{\infty }\frac{1}{m!}\left\| \sum_{i=1}^{m}\partial _{i}\Psi
_{m}\right\| ^{2}<\infty
\end{array}
\right\} ; \\
\mathscr{V}_{s} &=&\left\{ \Psi \in \mathscr{W}:\sum_{m=0}^{\infty }\frac{1}{%
m!}\left\| \Psi _{m+1}\left( \cdot
,t_{m+1}=s\right) \right\| ^{2}<\infty \right\} ; \\
\mathscr{V}_{0^{\pm }} &=&\mathscr{V}_{0^{+}}\cap \mathscr{V}_{0-}.
\end{eqnarray*}
We remark that $\mathscr{W}$ is the natural domain for $d\Gamma (i \partial)$. On $\mathscr{V}_s$ we define the operators
\begin{eqnarray*}
\left( a(s) \, \Psi \right) = \Psi_{n+1} (\cdot, t_{n+1}= s).
\end{eqnarray*}
On the subspace $\mathscr{V}_{0^\pm}$, the operators $d\Gamma (i \partial )$ and 
$a( 0^\pm )$ are all simultaneously defined.
 
\begin{definition}[The Gregoratti Hamiltonian]
\label{def Gregoratti Ham} Define the following operator $H$ on $\mathfrak{h}\otimes \mathcal{F}$ 
\begin{eqnarray}
H \Phi & =&  d\Gamma (i\partial _{ac}) \Phi 
-i\mathsf{L}^{\dag }\mathsf{S}a\left( 0^{+}\right) \Phi 
+\left( \mathsf{H}-\dfrac{i}{2}\mathsf{L}^{\dag }\mathsf{L}\right) \Phi, 
\label{eq:CGhamiltonian0} \\
\mbox{Dom}(H)& =&\left\{   \Phi \in \mathscr{V}_{0^{\pm }}: \
a(0^{-})\Phi =\mathsf{S} a (0^{+})\Phi +\mathsf{L} \Phi
\right\} .
\label{eq:CGhamiltonian}
\end{eqnarray}
\end{definition}

It follows from the work of Chebotarev and Gregoratti \cite{Che97,Gre01} that the operator $H$ is essentially self-adjoint and its unique self-adjoint extension generates the unitary group $U\left( t\right) =\Theta _{t}V_{t}$ where $V_{t}$ is the unitary solution to the following quantum stochastic differential equation (\ref{eq: diff eq}):

\begin{eqnarray}
dV\left( t\right) &=&\left\{ 
(\mathsf{S}-1)d\Lambda  ( t ) +\mathsf{L} dA^{\dag } ( t ) 
-\mathsf{L}^{\dag }\mathsf{S}dA  ( t ) -\dfrac{1}{2}\mathsf{L}^{\dag }\mathsf{L} dt-i\mathsf{H}dt
\right\} V\left( t\right) , \notag  \\
V\left( 0\right) &=&I.
\end{eqnarray}

The main result of this section is the following theorem.

\begin{theorem}
\label{thm main result} Let $0\leq T<\infty $. We have the following 
\begin{equation*}
\lim_{k\rightarrow \infty }\sup_{0\leq t\leq T}\left\| \left( U^{\left(
k\right) }\left( t\right) -U\left( t\right) \right) \Phi \right\| =0,\qquad
\forall \Phi \in \mathfrak{h}\otimes \mathcal{F}.
\end{equation*}
\end{theorem}

Before proving the theorem (see the end of this section), we make some preparations. As in the previous section, we would like to use the Trotter-Kato Theorem, therefore, for every $\Phi $ in a core for $\mbox{Dom}(H)$, we need to construct an approximating sequence $\Phi ^{(k)}$ that satisfies the first condition of Theorem \ref{thm Trotter-Kato}. We again employ a smearing through convolution with $g^{\left( k\right) }$, this time applied as a second quantization.

\begin{definition}
\label{def convolution} Let $g^{\left( k\right) }$ be as in Definition \ref{def:gk} and assume further that $g(t) \ge 0$ for all t (hence $\|g\|_1=1$).
 Let $G^{(k)}:L^{2}(\mathbb{R})\rightarrow L^{2}(\mathbb{R})$ be the convolution with $g^{(k)}$, i.e.\ 
\begin{equation*}
G^{(k)}h=g^{(k)}\ast h,\qquad \forall h\in L^{2}(\mathbb{R}).
\end{equation*}
Let $\Phi $ be an element in $\mbox{Dom}(H)$. We define an element $\Phi ^{(k)}$ in the domain of $H^{(k)}$ by 
\begin{equation}
\Phi ^{(k)}=\Gamma (G^{(k)})\Phi .  \label{eq: Phi k}
\end{equation}
Here $\Gamma (G^{(k)})$ denotes the second quantization of $G^{(k)}$. 
\end{definition}

Note that $G^{(k)}$ is a contraction ($\Vert g^{(k)}\Vert _{1} = 1$, i.e.\ $\Vert \hat{g}^{(k)}\Vert _{\infty }\leq 1$ with $\hat{g}^{(k)}$ the Fourier transform $\hat{g}^{(k)}=\int_{-\infty}^{\infty} g^{(k)}(t)e^{-i\omega t}dt$), so its second quantization is well-defined). The positivity assumption on $g$ implies that $\kappa_+ = \kappa_- = \frac{1}{2}$ (which agrees with Section 1.2). 

\begin{lemma}
\label{lem lieblossfock} For all $\Phi \in \mathfrak{h}\otimes \mathcal{F}$, we have 
\begin{equation*}
\lim_{k\rightarrow \infty }\Gamma (G^{(k)})\Phi =\Phi .
\end{equation*}
\end{lemma}

\begin{proof}
Since the linear span of exponential vectors $v\otimes e(h)$ is dense in $\mathfrak{h}\otimes \mathcal{F}$ and $\Gamma (G^{(k)})$ is bounded, it is enough to prove the Lemma for all vectors of the form $\Phi =v\otimes e(h)$. We have 
\begin{gather*}
\Vert \Gamma (G^{(k)})v\otimes e(h)-v\otimes e(h)\Vert ^{2}= \\
\Vert v\Vert ^{2}\left[ \exp (\Vert G^{(k)}h\Vert ^{2})+\exp (\Vert h\Vert
^{2})-\exp (\langle G^{(k)}h|h\rangle )-\exp (\langle h|G^{(k)}h\rangle )%
\right] \rightarrow 0,
\end{gather*}
where in the last step we used \cite[Thm.\ 2.16]{LiL97}.
\end{proof}

We now recall the following result, see for instance \cite{Petz}.

\begin{lemma}
\label{lem convann} Let $C:\ L^{2}(\mathbb{R})\rightarrow L^{2}(\mathbb{R})$ be a contraction. We have for $h\in L^{2}(\mathbb{R})$ 
\begin{equation*}
\Gamma (C)\Big( \mbox{Dom} (A(C^{\dag }h))\Big)\subset \mbox{Dom}(A(h)).
\end{equation*}
Moreover, on the domain of $A(C^{\dag }h)$, we have 
\begin{equation*}
A(h)\Gamma (C)=\Gamma (C)A(C^{\dag }h).
\end{equation*}
\end{lemma}

Note that we have the following second quantized version of equation (\ref{eq: deriv 1st}): 
\begin{equation*}
d\Gamma (i\partial )\Phi ^{\left( k\right) }=\Gamma (G^{(k)})d\Gamma
(i\partial _{ac})\Phi +iA^{\dag }(g^{\left( k\right) })\Gamma
(G^{(k)})a_{\jmath }\Phi ,
\end{equation*}
where 
\begin{equation*}
\left( a_{\jmath }\Phi \right) _{m}\left( t_{1},\cdots ,t_{m}\right) =\Phi
_{m+1}\left( t_{1},\cdots ,t_{m},0^{+}\right) -\Phi _{m+1}\left(
t_{1},\cdots ,t_{m},0^{-}\right) .
\end{equation*}
The action of $H^{\left( k\right) }$ on $\Phi ^{\left( k\right) }$ can now be written as 
\begin{eqnarray}
H^{\left( k\right) }\Phi ^{\left( k\right) } &=&\Gamma (G^{(k)})d\Gamma
\left( i\partial _{ac}\right) \Phi  \notag \\
&&+A^{\dag }(g^{\left( k\right) })\Gamma (G^{(k)})\left( i a_{\jmath }\Phi +
\mathsf{E}_{11}A(\rho ^{\left( k\right) })\Phi +\mathsf{E}_{10}\Phi \right) 
\notag \\
&&+\mathsf{E}_{01}\Gamma (G^{(k)})A(\rho ^{\left( k\right) })\Phi +\mathsf{E}_{00}\Gamma (G^{(k)})\Phi .  \label{eq:Hkfk}
\end{eqnarray}
Here we have used Lemma \ref{lem convann} and the fact that $A(G^{(k)\dag
}g^{(k)})=A(\rho ^{(k)})$.

\begin{lemma}
\label{lem:technical} The singular component of equation $\left( \ref{eq:Hkfk}\right) $ converges strongly to zero as $k\rightarrow \infty $, i.e., 
\begin{equation*}
\left\| A^{\dag }(g^{\left( k\right) })\Gamma (G^{(k)})\left( i a_{\jmath
}\Phi +\mathsf{E}_{11}A(\rho ^{\left( k\right) })\Phi +\mathsf{E}_{10}\Phi
\right) \right\| _{2}\overset{k\rightarrow \infty }{\longrightarrow }0,
\end{equation*}
for all $\Phi $ in a core domain $\mathcal{D}$ of $H$.
\end{lemma}

We defer the proof of this lemma to the next section.

Using Lemma \ref{lem lieblossfock}, we find that the first term in Equation \eqref{eq:Hkfk} converges to the first term in the Hamiltonian $H$ given by Equation \eqref{eq:CGhamiltonian0}, i.e.
\begin{equation*}
\lim_{k\rightarrow \infty }\left\| \Gamma (G^{(k)})d\Gamma \left( i\partial
_{ac}\right) \Phi -d\Gamma \left( i\partial _{ac}\right) \Phi \right\|
_{2}=0.
\end{equation*}
In the proof of the Lemma \ref{lem:technical}, it is shown that $A(\rho ^{(k)})\Phi $ converges in $L^{2}$-norm to $\frac{1}{2} a(0^{-})\Phi + \frac{1}{2} a(0^{+})\Phi $. Therefore, we find for the last line of Equation  (\ref{eq:Hkfk} )
\begin{equation*}
\mathsf{E}_{01}\Gamma (G^{(k)})A(\rho ^{\left( k\right) })\Phi +\mathsf{E}_{00}\Gamma (G^{(k)})\Phi \longrightarrow \mathsf{E}_{01}
\left( \frac{1}{2} a\left( 0^{+}\right) + \frac{1}{2} a\left( 0^{-}\right) \right) \Phi +\mathsf{E}_{00}\Phi .
\end{equation*}
Employing the boundary condition, we have that 
\begin{eqnarray*}
&&\mathsf{E}_{01}\left( \frac{1}{2} a\left( 0^{+}\right) + \frac{1}{2} a\left(
0^{-}\right) \right) \Phi +\mathsf{E}_{00}\Phi \\
&=&\mathsf{E}_{01}\left( \frac{1}{2} a\left( 0^{+}\right) \Phi + \frac{1}{2}
\left[ \mathsf{S}\, a\left(0^{+}\right) \Phi + \mathsf{L}\, \Phi \right] \right) +\mathsf{E}_{00}\Phi \\
&\equiv &-i\mathsf{L}^{\dag }\mathsf{S}a\left( 0^{+}\right) \Phi +(\mathsf{H} -\dfrac{i}{2}\mathsf{L}^{\dag }\mathsf{L})\Phi .
\end{eqnarray*}
Here we have used the algebraic identities 
\begin{equation*}
\mathsf{E}_{01}\left( \frac{1}{2} +\frac{1}{2} \mathsf{S}\right) 
=\mathsf{E}_{01}\left( \frac{1}{2}+\frac{1}{2}\dfrac{I-i\frac{1}{2}\mathsf{E}_{11}}{I+i\frac{1}{2}\mathsf{E}_{11}}\right) 
=\mathsf{E}_{01}  \frac{1}{I+i \frac{1}{2} \mathsf{E}_{11}}\equiv -iL^{\dag }\mathsf{S},
\end{equation*}

\begin{equation*}
-i\dfrac{\frac{1}{2}}{I+i \frac{1}{2}\mathsf{E}_{11}}=\dfrac{1}{2} \mathrm{Im} \, \left\{ 
\dfrac{\frac{1}{2}}{I+i\frac{1}{2}\mathsf{E}_{11}}\right\} -\dfrac{i}{2} \dfrac{I}{I+i\frac{1}{2}\mathsf{E}_{11}}\dfrac{I}{I-i\frac{1}{2}\mathsf{E}_{11}}.
\end{equation*}
Applying the Trotter-Kato Theorem, this completes the proof of our main result Theorem \ref{thm main result}.

\section{Proof of Lemma \ref{lem:technical}}

Setting $V^{\left( k\right) }=i a_{\jmath }\Phi +\mathsf{E}_{11}A(\rho
^{\left( k\right) })\Phi +\mathsf{E}_{10}\Phi $, we see that 
\begin{eqnarray*}
&&\left\| A^{\dag }(g^{\left( k\right) })\Gamma (G^{(k)})V^{\left( k\right)
}\right\| _{2}^{2} \\
&=&\left\langle \Gamma (G^{(k)})V^{\left( k\right) }\Big|\,A(g^{(k)})A^{\dag
}(g^{(k)})\Gamma (G^{(k)})V^{\left( k\right) }\right\rangle \\
&=&\left\langle \Gamma (G^{(k)})V^{\left( k\right) }\Big|\,\left( A^{\dag
}(g^{(k)})A(g^{(k)})+\Vert g^{(k)}\Vert _{2}^{2}\right) \Gamma
(G^{(k)})V^{\left( k\right) }\right\rangle \\
&\leq &\left\| A(g^{\left( k\right) })\Gamma (G^{(k)})V^{\left( k\right)
}\right\| _{2}^{2}+\Vert g^{(k)}\Vert _{2}^{2}\Vert V^{\left( k\right)
}\Vert _{2}^{2},
\end{eqnarray*}
where in the last step we used that $\Gamma (G^{(k)})$ is a contraction. We need to establish two further results:
the first is that  $V^{\left( k\right) }$ goes to $0$ sufficiently quickly and we prove this in Lemma
\ref{lem:pseudo} below; then we will have to show that this implies that the first term $\left\| A(g^{\left( k\right) })\Gamma (G^{(k)})V^{\left( k\right) }\right\| _{2}^{2}$ converges to $0$ and we prove this in Lemma \ref{lem:extra}. 

If we  accept these results for the moment, then from the boundary conditions we have 
\begin{eqnarray*}
&&ia_{\jmath }\Phi +\mathsf{E}_{11}\left( \frac{1}{2} a\left( 0^{+}\right)
+\frac{1}{2}a\left( 0^{-}\right) \right) \Phi +\mathsf{E}_{10}\Phi \\
&=&i\left( I-i\frac{1}{2}\mathsf{E}_{11}\right) a\left( 0^{+}\right) \Phi
-i\left( I+i\frac{1}{2} \mathsf{E}_{11}\right) a\left( 0^{-}\right) \Phi +%
\mathsf{E}_{10}\Phi \\
&=&i\left( I+i\frac{1}{2} \mathsf{E}_{11}\right) \left[ Sa\left( 0^{+}\right)
\Phi +L\Phi -a\left( 0^{-}\right) \Phi \right] =0
\end{eqnarray*}
so that, in fact, 
\begin{equation*}
V^{\left( k\right) }=\mathsf{E}_{11}\left[ A(\rho ^{\left( k\right) })\Phi
-\left( \frac{1}{2}a\left( 0^{+}\right) +\frac{1}{2}a\left( 0^{-}\right)
\right) \Phi \right]
\end{equation*}
As $\Vert g^{k}\Vert _{2}$ grows at rate $\sqrt{k}$, it suffices to show that $A(\rho ^{\left( k\right) })\Phi -(\frac{1}{2}a\left( 0^{+}\right) +\frac{1}{2}a\left( 0^{-}\right) )\Phi $ goes to $0$ in norm with rate faster than $\dfrac{1}{\sqrt{k}}$. We will now establish this result below, but first we need to recall the definition of a pseudo-exponential vector from \cite{Gre01}.

\begin{definition}
Let $F \colon t \mapsto \mathsf{F}_t$ be a function from $  \mathbb{R}$ to $\mathfrak{B}\left( \mathfrak{h}\right) $ and define the corresponding pseudo-exponential vector $\Psi \left( \mathsf{F},h\right) $ as 
\begin{equation*}
\left[ \Psi \left( \mathsf{F},h\right) \right] _{m}\left( t_{1},\cdots
,t_{m}\right) =\vec{T}\mathsf{F}_{t_{1}}\cdots \mathsf{F}_{t_{m}}h
\end{equation*}
for given $h\in \mathfrak{h}$, where $\vec{T}$ denotes chronological ordering. That is 
\begin{equation*}
\vec{T}\mathsf{F}_{t_{1}}\cdots \mathsf{F}_{t_{m}}=\mathsf{F}_{t_{\sigma
\left( 1\right) }}\cdots \mathsf{F}_{t_{\sigma \left( m\right) }}
\end{equation*}
where $\sigma $ is a permutation for which $t_{\sigma \left( 1\right) }\geq
\cdots \geq t_{\sigma \left( m\right) }$.
\end{definition}

\begin{lemma}
\label{lem:pseudo}
Let $v\in W^{1,2}\left( \mathbb{R}/\left\{ 0\right\}  \right) $ and $u\in
W^{1,2}\left( \mathbb{R}/\left\{ 0\right\} \right) $ with $\left. u\right| _{\mathbb{R}_{+}}=0$ and $u\left( 0^{-}\right) =1$, then define $\mathsf{F}_{t} $ by 
\begin{equation}
\mathsf{F}_t =v \left( t\right) +u\left( t\right) \left[
\mathsf{S} v \left( 0^{+}\right) +\mathsf{L} -v \left(
0^{-}\right) \right]
\label{eq:pseudovector_F}
\end{equation}
then the domain $\mathcal{D}$ of such pseudo-exponential vectors $\Phi =\Psi \left( \mathsf{F},h\right) $ is a core for $H$. Moreover, for each such vector we have 
\begin{equation*}
\left\| A(\rho ^{\left( k\right) })\Phi -\left( \frac{1}{2}a\left(
0^{+}\right) +\frac{1}{2}a\left( 0^{-}\right) \right) \Phi \right\|
_{2}=O\left( \frac{1}{k}\right) .
\end{equation*}
\end{lemma}

\begin{proof}
The first part of this lemma is proved by Gregoratti where it is shown that $\mathcal{D}$ is dense, and is contained in $\mbox{Dom} (H) \cap \mathscr{V}_{0^\pm}$, see \cite{Gre01} Propositions 4 and 5. Note that for $\Phi =\Psi \left( \mathsf{F},h\right) $, by (4) in \cite{Gre01} we have 
\begin{eqnarray*}
a \left( t\right) \Phi &=&v \left( t\right) \Phi ,\quad t\in \left\{
0^{+}\right\} \cup \left( 0,\infty \right) , \\
a \left( 0^{-}\right) \Phi &=&\left( \mathsf{S} v \left(
0^{+}\right) +\mathsf{L} \right) \Phi .
\end{eqnarray*}
To prove the second part, we begin by setting 
\begin{eqnarray*}
Z_{m}\left( t_{1},\cdots ,t_{m}\right)
&=& \left[ A (\rho^{( k ) }) \Phi -\left( \frac{1}{2} a  ( 0^{+} ) +\frac{1}{2}a  ( 0^{-} \right) 
\Phi \right] _{m} ( t_{1},\cdots ,t_{m} ) \\
&=& \int_0^\infty \rho^{ ( k ) } ( s ) 
\left[
\Phi _{m+1} ( t_{1},\cdots ,t_{m}, s  ) 
-\Phi_{m+1} ( t_{1},\cdots ,t_{m}, 0^{+}   ) \right] ds \\
&& +\int_{-\infty }^{0}\rho ^{ ( k ) } ( s ) \left[
\Phi _{m+1} ( t_{1},\cdots ,t_{m}, s   ) -\Phi_{m+1} ( t_{1},\cdots ,t_{m}, 0^{-}   ) \right] ds \\
&\equiv & Z_{m}^{+}\left( t_{1},\cdots ,t_{m}\right) +Z_{m}^{-}\left(
t_{1},\cdots ,t_{m}\right) .
\end{eqnarray*}
We have $\left\| Z_{m}\right\| ^{2}\leq \left( \left\| Z_{m}^{+}\right\|
+\left\| Z_{m}^{-}\right\| \right) ^{2}$ but 
\begin{equation*}
Z_{m}^{+}\left( t_{1},\cdots ,t_{m}\right) =\int_{0}^{\infty }\rho
^{\left( k\right) }\left( s\right) \left[ v \left( s\right)
-v \left( 0^{+}\right) \right] ds\,\Phi _{m}\left( t_{1},\cdots
,t_{m}\right)
\end{equation*}
and this prefactor is clearly $O\left( \dfrac{1}{k}\right) $ from the argument used in Lemma \ref{lem:TKconditions}.

However, we then have 
\begin{eqnarray*}
&&Z_{m}^{-}\left( t_{1},\cdots ,t_{m}\right) \\
&=&\int_{-\infty }^{0}\rho ^{\left( k\right) }\left( s\right)  
\left[ \mathsf{F}_{t_{\sigma \left( 1\right) }}\cdots \mathsf{F}_{s }\cdots \mathsf{F}_{t_{\sigma \left( m\right) }}-\mathsf{F}_{ 0^{-}  }\mathsf{F}_{t_{\sigma \left( 1\right) }}\cdots 
\mathsf{F}_{t_{\sigma \left( m\right) }}\right] h\,ds
\end{eqnarray*}
where $\sigma $ is the chronological time ordering permutation. 

We note however that  $[ \mathsf{F}_t , \mathsf{F }_s ]=0 $ for all $t,s$, therefore we have
\begin{eqnarray*}
Z_{m}^{-}\left( t_{1},\cdots ,t_{m}\right) 
&=&\int_{-\infty }^{0}\rho ^{\left( k\right) }\left( s\right)  
\left[   \mathsf{F}_{s } -\mathsf{F}_{ 0^{-}  }\right] 
\mathsf{F}_{t_{\sigma \left( 1\right) }}\cdots 
\mathsf{F}_{t_{\sigma \left( m\right) }}h \, ds \\
&=&\int_{-\infty }^{0}\rho ^{\left( k\right) }\left( s\right)  
\left[   u (s)  - u( 0^{-} ) \right] 
\left[ \mathsf{S} v(0^+) + \mathsf{L} - v(0^-) \right] \\
&& \times 
\mathsf{F}_{t_{\sigma \left( 1\right) }}\cdots 
\mathsf{F}_{t_{\sigma \left( m\right) }}h \, ds
\end{eqnarray*}
where we used (\ref{eq:pseudovector_F}). From the argument in Lemma \ref{lem:TKconditions} again, 
we see that this is $O\left( \dfrac{1}{k}\right) $.
\end{proof}

\begin{lemma}
\label{lem:extra}
For $\Phi$ chosen as a pseudo-exponential vector, as in Lemma \ref{lem:pseudo}, we have that $\left\| A(g^{\left( k\right) })\Gamma (G^{(k)})V^{\left( k\right) }\right\| _{2}^{2}$ converges to $0$ as $k \to \infty$.
\end{lemma}

\begin{proof}
We have that
\begin{eqnarray*}
 A(g^{\left( k\right) })\Gamma (G^{(k)})V^{\left( k\right) } =
  \Gamma (G^{(k)}) A( \rho^{\left( k\right) })V^{\left( k\right) },
\end{eqnarray*}
with $\Gamma (G^{(k)})$ a contraction. The $m$th level of the Fock space component of $A( \rho^{\left( k\right) })V^{\left( k\right) }$ may be written as
\begin{eqnarray*}
E_{11} A( \rho^{\left( k\right) })\, Z_m^+
+ 
E_{11} A( \rho^{\left( k\right) })\, Z_m^- ,
\end{eqnarray*}
where we use the same conventions as in Lemma \ref{lem:pseudo}.  The first term has the explicit components
\begin{eqnarray*}
E_{11} \int dt \, \rho^{(k)} (t) \int_{0}^{\infty }\rho
^{\left( k\right) }\left( s\right) \left[ v \left( s\right)
-v \left( 0^{+}\right) \right] ds\,\Phi _{m+1}\left( t, t_{1},\cdots 
,t_{m}\right) \\
=
E_{11} \int dt \, \rho^{(k)} (t)  \mathsf{F}_t \, \int_{0}^{\infty }\rho
^{\left( k\right) }\left( s\right) \left[ v \left( s\right)
-v \left( 0^{+}\right) \right] ds\,\Phi _{m }\left(   t_{1},\cdots 
,t_{m}\right)
\end{eqnarray*}
which is norm bounded by $ \| E_{11} \| \, \int dt \, \rho^{(k)} (t)  \| \mathsf{F}_t \| \, \| Z^+_m \|$, and we note that in fact $\int dt \, \rho^{(k)} (t)  \| \mathsf{F}_t \| =\int d\tau \, \rho (\tau)  \| \mathsf{F}_{\tau / k} \|$. An equivalent bound is easily shown to hold for $E_{11} A( \rho^{\left( k\right) })\, Z_m^-$ and so by an argument similar to lemma \ref{lem:pseudo} we obtain the desired result.
\end{proof}

\subsection*{Epilogue}
After completion of this work, the authors became aware of the book by W. von Waldenfels \cite{VW} which gives a
complete resolvent analysis of the Chebotarev-Gregoratti-von Waldenfels Hamiltonian, and in the final chapter 
describes a strong resolvent limit by colored noise approximations. The convergence is comparable to the
strong uniform convergence considered here, but the approach is very different.

\begin{acknowledgement}
JG and RG are grateful to EPSRC EP/G039275/1 and EP/L006111/1 grant for support. They also wish to thank the Isaac Newton Institute for Mathematical Sciences, 
Cambridge, for support and hospitality during the programme \textit{Quantum Control Engineering} 
where work on this paper was completed. HN acknowledges
support through a research visit funded through EPSRC EP/H016708/1 and Australian Research Council grants DP0986615 and DP130104191.
\end{acknowledgement}


\bigskip

\begin{thebibliography}{9}
\bibitem{Che97}  A.M. Chebotarev, \textit{Quantum stochastic differential
equation is unitarily equivalent to a symmetric boundary problem for the
Schr\"{o}dinger equation}, Math. Notes, \textbf{61}, No. 4, 510-518, (1997)

\bibitem{Go05}  J. Gough, \textit{Quantum flows as Markovian limit of
emission, absorption and scattering interactions}, Commun. Math. Phys., 
\textbf{254}, 489-512, (2005)

\bibitem{Go06}  J. Gough, \textit{Quantum Stratonovich Calculus and the
Quantum Wong-Zakai Theorem}, Journ. Math. Phys. \textbf{47}, 113509, (2006)

\bibitem{Gre01}  M. Gregoratti, \textit{The Hamiltonian operator associated
to some quantum stochastic differential equations}, Commun. Math. Phys., 
\textbf{222}, 181-200, (2001)

\bibitem{WvW} W. von Waldenfels, \textit{ Symmetric differentiation and Hamiltonian of a quantum stochastic process}, 
Infin. Dimens. Anal. Quantum Probab. Relat. Top. 8, no. 1, 73-116 (2005)

\bibitem{HP84}  R.L. Hudson and K.R. Parthasarathy, \textit{Quantum Ito's
formula and stochastic evolutions,} Commun. Math. Phys. \textbf{93}, 301-323
(1984)

\bibitem{AFL1} L. Accardi, A. Frigerio, Y.G. Lu, \textit{The weak coupling limit as
a quantum functional central limit}, Comm.Math.Phys. 131, 537-570 (1990)

\bibitem{AFL2} L. Accardi, A. Frigerio, Y.G. Lu, \textit{The low density limit in finite temperature case}, Nagoya Math. J. Volume 126, 25-87 (1992)

\bibitem{Da} E.B. Davies, \textit{One-Parameter Semigroups}, Academic Press (1980)

\bibitem{Gu} K.E. Gustafson, \textit{Introduction to Partial Differential Equations and Hilbert Space Methods}, 3rd ed., Dover Publications (1999)

\bibitem{LiL97}  E.H. Lieb, M. Loss, \textit{Analysis}, American
Mathematical Society, Providence, Rhode Island, (1997)

\bibitem{RSI}  M. Reed, B. Simon, \textit{Methods of Mathematical Physics I:
Functional Analysis}, Academic Press (1980)

\bibitem{RSII}  M. Reed, B. Simon, \textit{Methods of Mathematical Physics II:
Fourier Analysis, Self-adjointness}, Academic Press (1975)

\bibitem{QG}  R. Quezada-Batalla, O. Gonz\'{a}lez-Gaxiola, \textit{On the
Hamiltonian of a class of quantum stochastic processes}, Math. Notes, 
\textbf{81}, 5-6, 734-752, (2007)

\bibitem{Petz}  D. Petz, \textit{Invitation to the Canonical Commutation
relations}, Leuven University Press (1990)


\bibitem{VW}  W. von Waldenfels, \textit{A Measure Theoretical Approach to Quantum
Stochastic Processes}, Lecture Notes in Physics \textbf{878}, Springer (2014)



\end{thebibliography}
\end{document}